\let\origvec\vec
\documentclass{llncs}
\let\vec\origvec
\usepackage{amsmath}
\usepackage{amssymb}
\usepackage{framed}
\usepackage{graphicx}
\usepackage{etoolbox}
\usepackage{pdfpages}
\usepackage{caption}

\usepackage{amsthm}

\usepackage[utf8]{inputenc}
\usepackage[T1]{fontenc}
\usepackage{parskip} 
\usepackage{physics}

\usepackage[matrix,frame,arrow]{xy}
\usepackage{amsmath}
\newcommand{\qw}[1][-1]{\ar @{-} [0,#1]}
\newcommand{\multigate}[2]{*+<1em,.9em>{\hphantom{#2}} \qw \POS[0,0].[#1,0];p !C *{#2},p \save+LU;+RU **\dir{-}\restore\save+RU;+RD **\dir{-}\restore\save+RD;+LD **\dir{-}\restore\save+LD;+LU **\dir{-}\restore}
\newcommand{\ghost}[1]{*+<1em,.9em>{\hphantom{#1}} \qw}

\newcommand{\rstick}[1]{*!L!<-.5em,0em>=<0em>{#1}}
\newcommand{\lstick}[1]{*!R!<.5em,0em>=<0em>{#1}}

\newcommand{\Qcircuit}{\xymatrix @*=<0em>}

\spnewtheorem{protocol}[theorem]{Protocol}{\bfseries}{\itshape}

\usepackage{tikz}
\usepackage{tikzpeople}
\usepackage{stackengine}

\usepackage{setspace}
\usepackage{pdfpages}
\usepackage{marvosym}

\title{The RGB No-Signalling Game}

\author{
Xavier Coiteux-Roy\thanks{Supported in part by SNF and FRQNT.}
and
Claude Cr\'epeau\thanks{Supported in part by FRQNT (INTRIQ) and NSERC (CryptoWorks21 and Discovery grant program).}
}
\institute{
Universit\`a della Svizzera italiana, Lugano, Switzerland.
{xavier.coiteux.roy@usi.ch}\\
McGill University, Montr\'eal, Qu\'ebec, Canada.
{crepeau@cs.mcgill.ca}
}

\begin{document}

\maketitle


\addtocounter{footnote}{1}
\begin{abstract}

Introducing the simplest of all No-Signalling Games: the RGB Game where two verifiers interrogate two provers, Alice and Bob, far enough from each other that communication between them is too slow to be possible. Each prover may be independently queried one of three possible colours: Red, Green or Blue. Let $a$ be the colour announced to Alice and $b$ be announced to Bob. To win the game they must reply colours $x$ (resp. $y$) such that $a \neq x \neq y \neq b$.

This work focuses on this new game mainly as a pedagogical tool for its simplicity but also because it triggered us to introduce a new set of definitions for reductions among multi-party probability distributions and related {\em locality classes}. We show that a particular winning strategy for the RGB Game is equivalent to the PR-Box of Popescu-Rohrlich and thus No-Signalling. Moreover, we use this example to define No-Signalling in a new useful way, as the intersection of two natural classes of multi-party probability distributions called one-way signalling. We exhibit a quantum strategy able to beat the classical local maximum winning probability of 8/9 shifting it up to 11/12.
Optimality of this quantum strategy is demonstrated using the standard tool of semidefinite programming.

\end{abstract}


\section{The Game}\label{Game}

Claude started this research trying to find the simplest example he could think of to illustrate
multi-party distributions achievable via entanglement and No-Signalling in general.
His interest started from the following question on Quora:
``Could someone explain quantum entanglement to me like I'm 5 years old?'' 
Jon Hudson \cite{HUDSON18}, a former Stanford QM student, had given an answer involving friends choosing to have
pizza (or not) on the Moon and on Earth but he did not quite come up with a crisp No-Signalling
situation. Claude cooked up the RGB example after reading Jon's answer.

The canonical examples in this area are the Magic Square Game \cite{PhysRevLett.65.3373,PERES1990107} and the
so-called PR-box \cite{PR94} of Popescu-Rohrlich, both of which require some basic notions of arithmetics to
be introduced, or at least some basic logic as a common background. The purpose now is to present
an example so simple that even a five year old would understand it!

The RGB game is as follows:
\begin{quote}
``
Two people, Alice and Bob, play a game with friends Albert and Boris. Alice and Albert are on the moon,
while Bob and Boris stay on earth. Albert and Boris each independently picks at random a colour out of three possibilities:
Red, Green or Blue, and locally tells it to Alice or Bob.

Right away Alice and Bob choose a colour different from the one provided by their local counterpart.
For instance, if Albert tells Green to Alice, she may choose Red or Blue, while if Boris tells Red to Bob, he may choose Blue or Green.

Alice and Bob win the game if they never answer the same colour, either Red-Blue, Red-Green or Blue-Green in the example above.
''
\end{quote}

Figure 1 summarizes the input/output relation that Alice and Bob must satisfy.
$a$ is the colour given to Alice and $b$ is the colour given to Bob.
Their answers are $x$ and $y$ respectively. The condition they are trying to
achieve is $a \neq x \neq y \neq b$.

\begin{figure}[h!]

\centering
  \mbox{\Qcircuit @C=1em @R=.7em {
      \lstick{a}  \ar[r] & \multigate{1}{{\mathcal RGB}} & \rstick{b} \ar[l] \\
      \lstick{x}   & \ghost{{\mathcal RGB}} \ar[r]\ar[l] & \rstick{y}
    }}
\caption{an ${\mathcal RGB}$-box such that $a \neq x \neq y \neq b$}
  \label{RGB}
\end{figure}
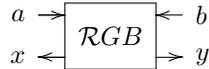

Such boxes are a standard way of representing the possible behaviours of Alice and Bob.
Indeed we can think of this box as a channel precisely describing the distribution
of $x,y$ given fixed values of $a,b$. The box of Figure 1 does not specify the
probabilities exactly and thus the name of the box is in calligraphic letters representing
the set of all the distributions that satisfy the given conditions. There are many distinct
ways of fulfilling the conditions of the game and many distributions that will win the
game 100\% of the time.

\subsection{Winning strategies}

Let's first consider a deterministic strategy for Alice and Bob's behaviour as described
by the box of Figure \ref{RGB0}.

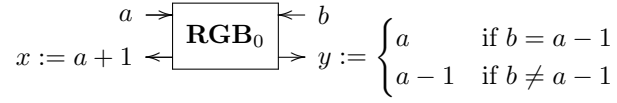
\begin{figure}[h!]

\centering
  \mbox{
  \Qcircuit @C=1em @R=.7em {
      \lstick{a}  \ar[r] & \multigate{1}{{\bf RGB}_0} & \rstick{b} \ar[l] \\
      \lstick{x:=a+1}   & \ghost{{\bf RGB}_0} \ar[r]\ar[l] & \rstick{y:= \begin{cases}
    a     & \text{if $ b  =    a-1 $} \\
    a-1  & \text{if $ b\neq a-1 $}
\end{cases} }
    }
}
\vspace{0,2cm}
\caption{a deterministic ${\bf RGB}_0$-box}
  \label{RGB0}
\end{figure}
In  this example we assume the colours are labelled $0,1$ or $2$ and that arithmetic operations are performed
modulo 3. 
When $a$ and $b$ are the same colour $u$ it produces
$$a=u, ~x=u+1, ~y=u-1, ~b=u.$$
The values $u+1$ and $u-1$ are the other two colours, distinct from $u$.
However, when $a$ and $b$ are the distinct colours $u,v$ it produces either
$$a=u, ~x=u+1, ~y=u, ~b=v$$
when the third colour is $u+1 = v-1$ or
$$a=u, ~x=u+1, ~y=u-1, ~b=v$$
when the third colour is $u-1 = v+1$.

This deterministic strategy defines completely the probability distribution of
the outputs $x,y$ given $a,b$: probability of $(x,y|a,b)$ is zero except when
$~x=a+1$ and $~y = \begin{cases}
    a     & \text{if $ b  =    a-1 $} \\
    a-1  & \text{if $ b\neq a-1 $}
\end{cases}
$ in which case it is precisely one.
Therefore we name this box ${\bf RGB}_0$ with bold characters because it precisely
defines a unique probability distribution $P_{x,y|a,b}$. This box achieves the prescribed
condition $a \neq x \neq y \neq b$ in a unique deterministic way for each $a,b$.

After complete examination of this condition one realizes that when $a=b$ is a single
colour $u$ the conditions can be satisfied in exactly two ways
$$a=u, ~x=u\pm 1, ~y=u\mp 1, ~b=u$$
whereas when $a$ and $b$ are distinct
colours $u,v$ the conditions can be satisfied in exactly three ways
$$a=u, ~x=v, ~y=u, ~b=v$$
$$a=u, ~x=u\pm 1, ~y=v\pm 1, ~b=v.$$

From this we conclude that out of the 9 possible $a,b$ pairs, three of them
($a=b$) may have two solutions and six of them
($a\neq b$) may have three solutions. This yields a total of $2^{3} 3^{6} = 18^{3} = 5832$ distinct deterministic winning
strategies. The above ${\bf RGB}_0$ strategy is only one of these.

We can completely parametrize all the winning strategies as a function of 15 real parameters
$p_{0}, p_{1}, p_{2}, p_{01}, p_{02}, p_{10}, p_{12}, p_{20}, p_{21}, q_{01}, q_{02}, q_{10}, q_{12}, q_{20}, q_{21}$ in the interval $[0,1]$ such that $p_{uv}+q_{uv}\leq 1$
as follows
\begin{equation}\label{uu}P_{u+1,u-1|u,u} = p_{u} \text{ and } ~P_{u-1,u+1|u,u} = 1-p_{u}, ~ \text{for } u\in \{ 0,1,2 \}\end{equation}
\begin{equation}\label{uv}P_{w,u|u,v} = p_{uv}, ~P_{v,w|u,v} = q_{uv} \text{ and } ~P_{v,u|u,v} = 1-p_{uv}-q_{uv}, ~ \text{for } \{ u,v,w \} = \{ 0,1,2 \}.\end{equation}
All the winning strategies to this game are among these probability distributions.
They are all the valid convex combinations of the $5832$ distinct deterministic winning
strategies.

The deterministic strategy ${\bf RGB}_0$ of Figure \ref{RGB0} is the special case
$$p_{0}=p_{1}=p_{2}=p_{02}=p_{20}=q_{01}=q_{10}=q_{12}=q_{21}=1, \text{ and } p_{01}=p_{10}=p_{12}=p_{21}=q_{02}=q_{20}=0.$$

\newcommand{\RGRB}[0]{ \text{$\mathbf{R  \! \frac{GR}{BG} \! B}$} }

The rest of this paper is going to focus on exactly one of these strategies with a very remarkable property:
it {\em does not require} Alice and Bob to signal to implemented it (whereas all the others actually do).
This strategy is going to be named
$\RGRB$\footnote{The name is a reminder that this strategy has the feature that whenever $a$ and $b$ are distinct, $axyb$ is $abcb$ or $acab$ ($c$ being the third colour) but never $abab$.
\RGRB is a combined string of types ${a  \! \frac{ca}{} \! b}$, ${a  \! \frac{}{bc} \! b}$. }
and is specified by the parameters
$$p_{0}=p_{1}=p_{2}=p_{01}=p_{10}=p_{02}=p_{20}=p_{12}=p_{21}=q_{01}=q_{10}=q_{02}=q_{20}=q_{12}=q_{21}=\frac{1}{2}.$$

\begin{figure}[h!]
\centering
  \mbox{\Qcircuit @C=1em @R=.7em {
      \lstick{a}  \ar[r] & \multigate{1}{ {{ \RGRB}} } & \rstick{b} \ar[l] \\
      \lstick{x}   & \ghost{{\RGRB}} \ar[r]\ar[l] & \rstick{y}
    }}
\caption{The ${{\RGRB}}$-box such that $a \neq x \neq y \neq b$, \newline \!\!\! and $(x,y) \neq (b,a)$, uniformly among solutions}
  \label{RGB}
\end{figure}
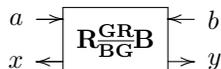

In Figure 3, ${{\RGRB}}$ is made precise by enforcing extra conditions on top of $a \neq x \neq y \neq b$. We force $P_{v,u|u,v} = 0$ by
adding $(x,y) \neq (b,a)$. Uniformity finally imposes that all the remaining non-zero probabilities be exactly $\frac{1}{2}.$

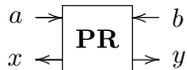
\begin{figure}[h!]
\centering
  \mbox{\Qcircuit @C=1em @R=.7em {
      \lstick{a}  \ar[r] & \multigate{1}{{\bf PR}} & \rstick{b} \ar[l] \\
      \lstick{x}   & \ghost{{\bf PR}} \ar[r]\ar[l] & \rstick{y}
    }}
\caption{a {\bf PR}-box satisfying the CHSH condition,\newline that $a \wedge b = x \oplus y$, uniformly among solutions}
  \label{nlbox}
\end{figure}

\subsection{Our Results}
The contributions of the paper are
\begin{enumerate}

\item Novel notion of reducibility among strategies
\item Novel definitions of basic notions such as locality, signalling, one-way signalling and no-signalling
\item A proof that our notion of no-signalling is equivalent to the generally accepted one
\item A proof of equivalence between ${{\RGRB}}$ and the well-known Popescu-Rohrlich Non-Local (yet No-Signalling) {\bf PR}-box (see Figure 4) Implying ${{\RGRB}}$ is also complete for the set of No-Signalling distributuions
\item A proof that ${{\RGRB}}$ is the ONLY No-Signalling distribution winning the RGB game
\item Quantum strategy with winning probability 11/12, better than the best local strategy at 8/9
\item A proof of optimality of this quantum strategy using semidefinite programming
\item Some related open problems

\end{enumerate}


\section{Definitions}
In this section we solely focus on the two-party single-round games and strategies that are sufficient
to discuss and analyze the strategies for the RGB game.
Definitions and proofs for complete generalizations to multi-party multi-round games and strategies
will appear in a forthcoming paper with co-authors Adel Magra and Nan Yang.

\subsection{Strategies: Two-Party Channels}

\subsubsection{Games:}

Let $V$ be a predicate on $A\times B\times X\times Y$ (for some finite sets $A, B, X,$ and $Y$) and let $\pi$ be a probability distribution on $A\times B$.
Then $V$ and $\pi$ define a (single-round) game $G$ as follows: A pair of questions $(a,b)$ is randomly chosen according to distribution $\pi$, and $a\in A$ is sent to Alice
and $b\in B$ is sent to Bob. Alice must respond with an answer $x\in X$ and Bob with an answer $y\in Y$.
Alice and Bob win if $V$ evaluates to 1 on $(a,b,x,y)$ and lose otherwise.

\subsubsection{Strategies:} A strategy for Alice and Bob is simply a probability distribution $P_{(x,y | a,b)}$ describing exactly how they will answer $(x,y)$ on every
pair of questions $(a,b)$. We now breakdown the set of all possible strategies for Alice and Bob according to their {\em locality}.

\subsubsection{Deterministic and Local Strategies:} A strategy $P_{(x,y | a,b)}$ is {\em deterministic} if there exists functions $f_{A}:A\rightarrow X, f_{B}:B\rightarrow Y$ such that
$$P_{(x,y | a,b)} = 
\begin{cases}
    1  & \text{if $x=f_{A}(a)$ and $y=f_{B}(b)$} \\
    0  & \text{otherwise}
\end{cases}.$$
A deterministic strategy corresponds to the situation where Alice and Bob agree on their individual actions before any knowledge of the values $a,b$ is provided to them.
In this case they use only their own input to determine their individual output.

A strategy $P_{(x,y | a,b)}$ is {\em local} if there exists a finite set $R$ and functions $f_{A}:A\times R\rightarrow X, f_{B}:B\times R\rightarrow Y$ such that
$$P_{(x,y | a,b)} = \frac{ |\{ r\in R : x=f_{A}(a,r) \text{ and } y=f_{B}(b,r) | }{ |R|}.$$
A local strategy corresponds to the situation where Alice and Bob agree on a deterministic strategy selected uniformly among $|R|$ such possibilities.
The choice $r$ of Alice and Bob's strategy, and the choice of inputs $(a,b)$ provided to Alice and Bob are generally agreed to be statistically independent random variables.

\subsection{Local Reducibility}
We now turn to the notion of locally reducing a strategy to another, that is how Alice and Bob limited to local strategies but equipped
with a particular (not necessarily local) strategy $U'$ are able to achieve another particular (not necessarily local) strategy $U$. For this purpose
we introduce a notion of distance between strategies in order to analyze strategies that are approaching each other asymptotically.

\subsubsection{Distances between Strategies:} Several distances could be selected here as long as their meaning as it approaches zero are the same. In the definitions below,
$U,U'$ are strategies and ${\mathcal U'}$ is a finite set of strategies. 

\begin{definition}
$ |U,U'| = \sum_{a,b,x,y} | P_{U}(x,y | a,b) - P_{U'}(x,y | a,b) | $
\end{definition}

\begin{definition}
$ |U,{\mathcal U'}| = \min_{U'\in {\mathcal U'} } |U,U'| $
\end{definition}

\subsubsection{Local extensions of Strategies:} For natural integer $n$, we define the set $\mbox{LOC}^{n}(U)$ of strategies that are local extensions (of order $n$)
of $U$ to be all the strategies Alice and Bob can achieve using local strategies where strategy $U$ may be used up to $n$ times as sub-routine
calls\footnote{This is done by selecting functions $f_{A}^{1}:A\times R\rightarrow A, ~f_{A}^{2}:A^{2}\times R\rightarrow A,..., ~f_{A}^{n}:A^{n}\times R\rightarrow X$ to determine the input of each sub-routine call based on input $a$ and previous inputs/outputs.}.

\begin{definition}
We say that $U'$ Locally Reduces to $U$ ($U' \leq_{\mbox{\scriptsize\em LOC}} U$) iff $\lim_{n\rightarrow \infty} | U' , \mbox{\em LOC}^{n}(U) | =0$.
\end{definition}

\begin{definition}
We say that $U'$ is Locally Equivalent to $U$ ($U' =_{\mbox{\scriptsize\em LOC}} U$) iff $U' \leq_{\mbox{\scriptsize\em LOC}} U$ and $U \leq_{\mbox{\scriptsize\em LOC}} U'$.
\end{definition}

Note: a similar notion of reducibility has been previously defined by Dupuis, Gisin, Hasidim, M\'ethot, and Pilpel \cite{dghmp07} but without taking the limit to infinity. In their model they have previously
showed that $n$ instances of the PR-box modulo $p$ cannot be used to replicate exactly the PR-box modulo $q$, for any distinct primes $p,q$. 
However, Forster and Wolf \cite{PhysRevA.84.042112} have previously proved that {\bf PR} is complete for No-Signalling distributions under a similar (asymptotic) definition.

\subsection{Locality}
We now define the lowest of the locality classes ${\mathbb{LOC}}$. We could define it directly from the notion of local strategies as defined above, but for analogy with the other classes
we later define, ${\mathbb{LOC}}$ is defined as all those strategies locally reducible to a {\em complete} strategy we call $\mathbf{ID}$ (see Figure \ref{ID}). Of course, any strategy is complete for this class.
\begin{figure}[h!]

\centering
  \mbox{\Qcircuit @C=1em @R=.7em {
      \lstick{a}  \ar[r] & \multigate{1}{\mathbf{ID}} & \rstick{b} \ar[l] \\
      \lstick{a}   & \ghost{\mathbf{ID}} \ar[r]\ar[l] & \rstick{b}
    }}
\caption{an $\mathbf{ID}$-box}
  \label{ID}
\end{figure}

\begin{definition}
${\mathbb{LOC}} = \{ U | U \leq_{\mbox{\scriptsize\em LOC}} \mathbf{ID} \}$
\end{definition}

Note: ${\mathbb{LOC}}$ is the class of strategies that John Bell \cite{BELL64} considered as classical hidden-variable theories that he compared to entanglement.
It is also the class of strategies that BenOr, Goldwasser, Kilian and Wigderson \cite{BGKW88} chose to define classical Provers in Multi-Provers Interactive Proof Systems.

\subsection{One-Way Signalling}
We now turn to One-Way Signalling which allows communication from one side to the other. We name the directions arbitrarily Left and Right.
We define ${\mathbf R} \text{-} {\mathbb{SIG}}$ (resp. ${\mathbf L} \text{-} {\mathbb{SIG}}$) as all those strategies locally reducible to a {\em complete} strategy we call ${\mathbf R} \text{-} {\mathbf{SIG}}$ (see Figure \ref{rSIG})
(resp. ${\mathbf L} \text{-} {\mathbf{SIG}}$ (see Figure \ref{lSIG})). These classes are useful to define what it means for a strategy to {\em signal} as well as the notion of {\em No-Signalling} strategies.

\begin{figure}[h!]

\centering
  \mbox{\Qcircuit @C=1em @R=.7em {
      \lstick{a}  \ar[r] & \multigate{1}{ {\mathbf R} \text{-} {\mathbf{SIG}} } & \rstick{b} \ar[l] \\
      \lstick{a}   & \ghost{ {\mathbf R} \text{-} {\mathbf{SIG}} } \ar[r]\ar[l] & \rstick{a}
    }}
\caption{an $\mathbf{R} \text{-} {\mathbf{SIG}}$-box}
  \label{rSIG}
\end{figure}

\begin{definition}
${\mathbf R} \text{-} {\mathbb{SIG}} = \{ U | U \leq_{\mbox{\scriptsize\em LOC}} {\mathbf R} \text{-} {\mathbf{SIG}} \}$
\end{definition}

\begin{definition}
We say that $U$ Right Signals (is ${\mathbf R} \text{-} {\mathbb{SIG}}$-verbose\footnote{We define the notion of $\mathbb{L}$-verbose in analogy to $\mathbb{NP}$-hard: it means ``as verbose as any distribution in locality class
$\mathbb{L}$''. In consequence, a distribution $U$ is $\mathbb{L}$-complete if $U \in \mathbb{L}$ and $U$ is $\mathbb{L}$-verbose.}) iff ${\mathbf R} \text{-} {\mathbf{SIG}} \leq_{\mbox{\scriptsize\em LOC}} U$.
\end{definition}

\begin{figure}[h!]

\centering
  \mbox{\Qcircuit @C=1em @R=.7em {
      \lstick{a}  \ar[r] & \multigate{1}{{\mathbf L} \text{-} {\mathbf{SIG}}} & \rstick{b} \ar[l] \\
      \lstick{b}   & \ghost{{\mathbf L} \text{-} {\mathbf{SIG}}} \ar[r]\ar[l] & \rstick{b}
    }}
\caption{an ${\mathbf L} \text{-} {\mathbf{SIG}}$-box}
  \label{lSIG}
\end{figure}

\begin{definition}
${\mathbf L} \text{-} {\mathbb{SIG}} = \{ U | U \leq_{\mbox{\scriptsize\em LOC}} {\mathbf L} \text{-} {\mathbf{SIG}} \}$
\end{definition}

\begin{definition}
We say that $U$ Left Signals (is ${\mathbf L} \text{-} {\mathbb{SIG}}$-verbose) iff ${\mathbf L} \text{-} {\mathbf{SIG}} \leq_{\mbox{\scriptsize\em LOC}} U$.
\end{definition}

\begin{definition}
We say that $U$ Signals iff $U$ Right Signals or Left Signals.
\end{definition}

We prove a first result that is intuitively obvious. We show that the complete strategy ${\mathbf R} \text{-} {\mathbf{SIG}}$ cannot be approximated in ${\mathbf L} \text{-} {\mathbb{SIG}}$
and the other way around.

\begin{theorem}\label{LR-impossible}
${\mathbf R} \text{-} {\mathbf{SIG}} \not\in {\mathbf L} \text{-} {\mathbb{SIG}}$ and
${\mathbf L} \text{-} {\mathbf{SIG}} \not\in {\mathbf R} \text{-} {\mathbb{SIG}}$.
\end{theorem}

\begin{proof}
Follows from a simple capacity argument. For all $n$, all the channels in $\mbox{LOC}^{n}({\mathbf R} \text{-} {\mathbf{SIG}})$ have zero left-capacity, while ${\mathbf L} \text{-} {\mathbf{SIG}}$ has non-zero left-capacity. And vice-versa.
\end{proof}

\subsection{Signalling}
We are now ready to define the largest of the locality classes ${\mathbb{SIG}}$. Indeed every possible strategy is in ${\mathbb{SIG}}$.

\begin{definition}

${\mathbb{SIG}} = \{ U | U \leq_{\mbox{\scriptsize\em LOC}} \mathbf{SIG}\}$
\end{definition}

\begin{figure}[h!]

\centering
  \mbox{\Qcircuit @C=1em @R=.7em {
      \lstick{a}  \ar[r] & \multigate{1}{{\mathbf{SIG}}} & \rstick{b} \ar[l] \\
      \lstick{b}   & \ghost{{\mathbf{SIG}}} \ar[r]\ar[l] & \rstick{a}
    }}
\caption{a ${\mathbf{SIG}}$-box}
  \label{SIG}
\end{figure}

\begin{definition}
We say that $U$ Fully Signals (is ${\mathbb{SIG}}$-verbose) iff $U$ Right Signals and Left Signals.
\end{definition}

\subsection{No-Signalling}
We finally define the less intuitive locality class ${\mathbb{NOSIG}}$ in relation to classes defined above. 

\begin{definition}
$\mathbb{NOSIG} = {\mathbf R} \text{-} {\mathbb{SIG}} \bigcap {\mathbf L} \text{-} {\mathbb{SIG}}$.
\end{definition}

A similar characterization may be found in \cite{Acin2015} Section 3 and \cite{Barnum05} Corollary 3.5.

\begin{theorem} \label{NSig}.
The above definition of $\mathbb{NOSIG}$ exactly coincides with the {\em traditional} notion of No-Signalling \cite{BLM+05}.
\end{theorem}

\begin{proof}
If $U$ is signalling then it is verbose for at least one of ${\mathbf R} \text{-} {\mathbb{SIG}}$ or ${\mathbf L} \text{-} {\mathbb{SIG}}$.
Without loss of generality, assume it is verbose for ${\mathbf R} \text{-} {\mathbb{SIG}}$. Then by theorem \ref{LR-impossible}, $U \not\in {\mathbf L} \text{-} {\mathbb{SIG}}$, thus $U \not\in {\mathbf R} \text{-} {\mathbb{SIG}} \bigcap {\mathbf L} \text{-} {\mathbb{SIG}}$.


If $U$ is no-signalling then Alice's marginal distribution is independent from Bob's input $b$. Therefore, she can sample an output $x$ according to her input $a$ only as $P_{X|A=a}$ deduced from $P_{X,Y|A,B}$.
Alice can now communicate $a,x$ to Bob. Bob given $a,b,x$ can select $y$ according to the distribution $P_{Y|A=a,B=b,X=x}$ deduced from $P_{X,Y|A,B}$. The produced $x,y$ will have distribution $P_{X,Y|A=a,B=b}$ as expected.
This proves $U \in {\mathbf R} \text{-} {\mathbb{SIG}}$. Membership to ${\mathbf L} \text{-} {\mathbb{SIG}}$ is proven similarly.

This completes the proof of Theorem \ref{NSig}.

\end{proof}

Figure \ref{HIER} shows the relation of these classes as well as the case obtained via quantum entanglement (${\mathbb{|LOC\rangle }}$) as considered by Bell \cite{BELL64}.

\begin{figure}[hbt]
\begin{center}
\fbox{
\includegraphics[width=1.0\textwidth]{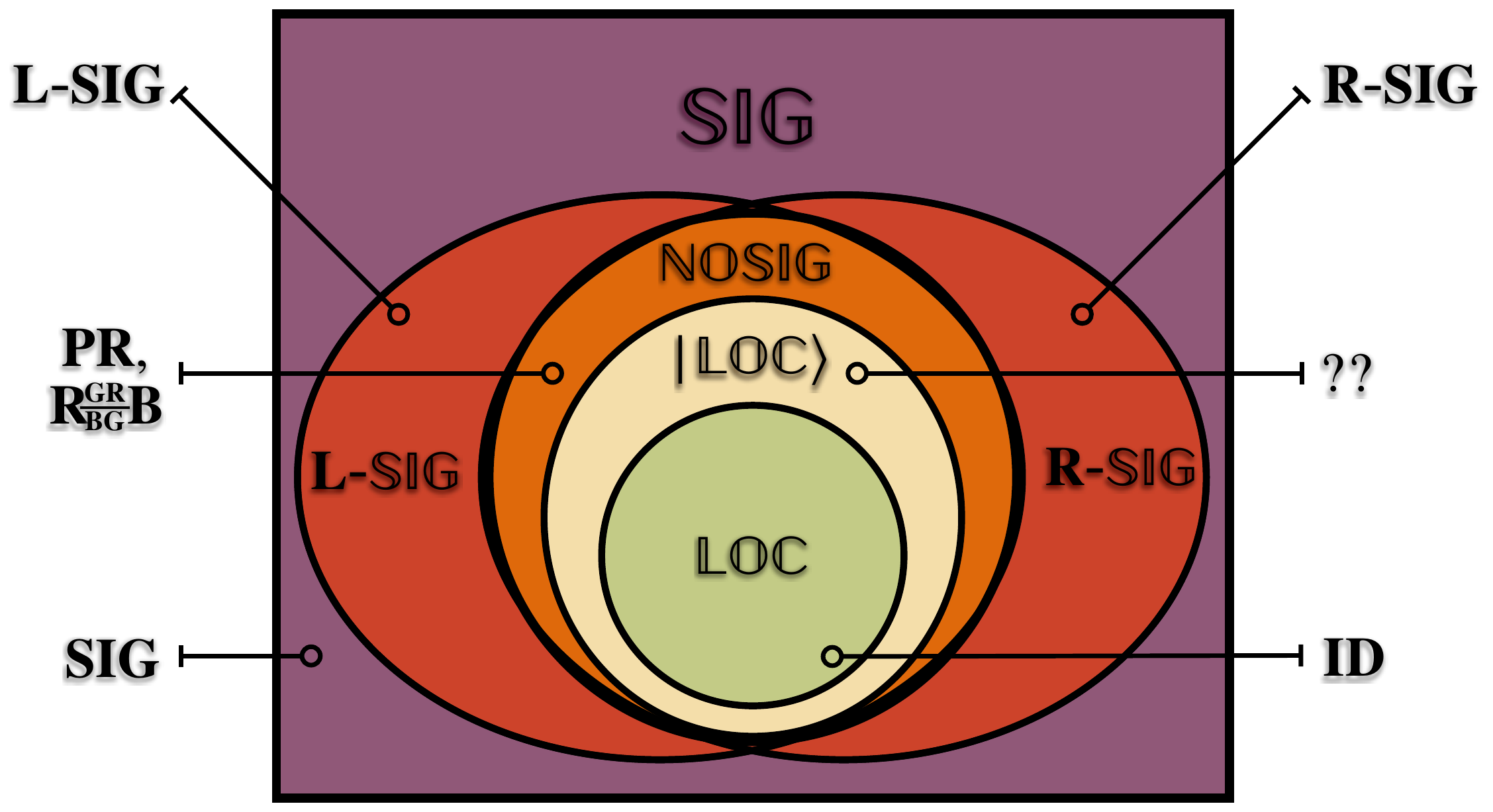}
}
\end{center}
\caption{\label{HIER} Locality Hierarchy and complete distributions in each class.}
\end{figure}

\begin{definition}
We say that $U$ does not Signal iff $U$ does not Right Signal nor Left Signal iff $U \in \mathbb{NOSIG}$.
\end{definition}

\begin{theorem}
If $U \in {\mathbf R} \text{-} {\mathbb{SIG}}$ (or $U \in {\mathbf L} \text{-} {\mathbb{SIG}}$) and $U$ is symmetric then $U$ does not Signal.
\end{theorem}

\begin{proof}
$U \in {\mathbf R} \text{-} {\mathbb{SIG}}$ and $U$ is symmetric imply that $U \in {\mathbf L} \text{-} {\mathbb{SIG}}$ as well. Thus $U\in {\mathbf R} \text{-} {\mathbb{SIG}} \bigcap {\mathbf L} \text{-} {\mathbb{SIG}}$.
\end{proof}


\begin{theorem}
${{\RGRB}} \in \mathbb{NOSIG}$
\end{theorem}

\begin{proof}
${{\RGRB}} \in {\mathbf R} \text{-} {\mathbb{SIG}}$ and ${{\RGRB}}$ is symmetric.
\end{proof}

\begin{theorem}\label{RGRBPR}
${{\RGRB}} =_{\mbox{\scriptsize\em LOC}} {\bf PR}$
\end{theorem}

\begin{proof}
First we show how ${\bf PR}$ may be achieved from ${{\RGRB}}$, more precisely that ${\bf PR}\in \mbox{LOC}^{1\!}({{\RGRB}})$.
All arithmetic operations are performed modulo 3.
Let $a' := f_{A}^{1}(a) := a$, and $b' := f_{B}^{1}(b) := 2b$.
The possible pairs for $(a',b')$ are therefore $(0,0), (0,2), (1,0), (1,2)$.
Let $(x',y') \leftarrow {{\RGRB}}(a',b')$.
Let $x := f_{A}^{2}(a,x') := 2(x'-a+1)$, and $y := f_{B}^{2}(b,y') := 2(y'-2b+1)$.
We leave it as an exercice to check that $(x,y)$ indeed satisfy the CHSH condition
that $x\oplus y = a \wedge b$.

Secondly, we show how ${{\RGRB}}$ may be achieved from ${\bf PR}$, more precisely that ${{\RGRB}}\in \mbox{LOC}^{2\!}({\bf PR})$.
Again, all arithmetic operations are performed modulo 3\footnote{Therefore modulo 2 for the exponents according to Fermat's little theorem.}.
The intuition in this case is that if $a=b$ then $(x,y)$ should be either $(a+1,b-1)$ or $(a-1,b+1)$ at random.
If $a\neq b$ then $(x,y)$ should be either $(a+1,b+1)$ or $(a-1,b-1)$ at random. 
The following computations achieve precisely this using the identity $a = b \mbox{ iff } (\neg a' \oplus b') \land (\neg a'' \oplus b'')$.

The first pair of functions compute the negation of the most significant bit of their inputs:
let $a' := f_{A}^{1}(a) := 1-2(a-1)a$, and $b' := f_{B}^{1}(b) := 1-2(b-1)b$.
Let $(x',y') \leftarrow {{{\bf PR}}}(a',b')$.

The second pair of functions compute the negation of the least significant bit of their inputs:
let $a'' := f_{A}^{2}(a,x') := 1-2(a-1)(a+1)$, and $b'' := f_{B}^{2}(b,y') := 1-2(b-1)(b+1)$.
Let $(x'',y'') \leftarrow {{{\bf PR}}}(a'',b'')$.

The third pair of functions compute $a\pm 1, b\pm 1$ according to the intuitive rule above:
let $x := f_{A}^{3}(a,x',x'') := a+2^{a_{1}*a_{2}+x'+x''}$, and $y := f_{B}^{3}(b,y',y'') := b+2^{b_{1}*b_{2}+y'+y''}$.
\end{proof}

\begin{corollary}\label{FW11}
{\RGRB} is $\mathbb{NOSIG}$-Complete.
\end{corollary}

\begin{proof}
Since {\bf PR} was previously proved $\mathbb{NOSIG}$-Complete by Forster and Wolf \cite{PhysRevA.84.042112} , then so is {\RGRB}.
\end{proof}

\begin{theorem}
${{\RGRB}}$ is the ONLY strategy winning the RGB game that is also No-Signalling.
\end{theorem}

\begin{proof}
%

Using the notation of Equations (\ref{uu}) -- (\ref{uv}), for No-Signalling on Alice's side we need
$$P_{u+1,u-1|u,u} = p_{u} = P_{u+1,u-1|u,u+1} + P_{u+1,u|u,u+1} = 1-p_{u\,u+1} = P_{u+1,u|u,u-1} = p_{u\,u-1}, 0\leq u\leq 2$$
%
and symmetrically on Bob's side
$$P_{u-1,u+1|u,u} = 1-p_{u} = P_{u-1,u+1|u+1,u} + P_{u,u+1|u+1,u} = 1-q_{u+1\,u} = P_{u,u+1|u-1,u} = q_{u-1\,u}, 0\leq u\leq 2.$$

Using all 6 sets of equalities we can get rid of all the variables but $p_{0}, p_{1}, p_{2}$ by setting
$$ p_{u\,u-1} = q_{u+1\,u} = p_{u}  \text{ and } p_{u\,u+1} = q_{u-1\,u} = 1-p_{u}, 0\leq u\leq 2.$$
It follows that
$$P_{u+1,u|u,u+1} = p_{u}+p_{u+1}-1 = -P_{u,u+1|u+1,u}, 0\leq u\leq 2$$
and since both $P_{u+1,u|u,u+1}$ and $P_{u,u+1|u+1,u}$ must be greater or equal to zero we conclude
$$P_{u+1,u|u,u+1} = P_{u,u+1|u+1,u} = 0  \text{ and }  p_{u} = 1-p_{u+1}, 0\leq u\leq 2.$$
It results that $p_{0} = 1-p_{1} = p_{2} = 1-p_{0} = p_{1} = 1-p_{2}$ and thus
$$p_{0} = p_{1} = p_{2} = p_{01} = p_{10} = p_{12} = p_{21} = p_{20} = p_{02} = q_{01} = q_{10} = q_{12} = q_{21} = q_{20} = q_{02} = 1/2$$
is the only solution as claimed.
%
%
%
%
\end{proof}

\begin{theorem}\label{RGB89}
The maximum local winning probability $p^{\operatorname{win}}_{\operatorname{local}}$ to the RGB game is 8/9.
\end{theorem}

\begin{proof}
Consider $f(R)=B$ and $f(G)=f(B)=R$ as well as $g(R)=g(B)=G$ and $g(G)=B$. By inspection of these functions we conclude $p^{\operatorname{win}}_{\operatorname{deterministic}}\geq 8/9$ because for all inputs $a,b$ we have $f(a)\neq a$ and $g(b)\neq b$ and 8 out of 9 input pairs $(a,b)$ are
such that $f(a)\neq g(b)$. 
Since it is a well known fact that $p^{\operatorname{win}}_{\operatorname{local}}=p^{\operatorname{win}}_{\operatorname{deterministic}}$, it suffices to show that $p^{\operatorname{win}}_{\operatorname{deterministic}} \leq 8/9$ as well.

To prove this, consider any pair of functions $f,g$. To obtain $f(a)\neq a$ for all $a$, the image of $f$ must contain at least 2 colours. Similarly for the image of $g$. Since both $f$ and $g$ can only take 3 values, their images
must have a common colour. Therefore, there exists an $a$ and a $b$ such that $f(a)=g(b)$. We conclude $p^{\operatorname{win}}_{\operatorname{deterministic}}\leq 8/9$, and therefore $p^{\operatorname{win}}_{\operatorname{local}}=p^{\operatorname{win}}_{\operatorname{deterministic}}=8/9$.
\end{proof}

Note: somewhat surprisingly Theorem \ref{RGRBPR} is not good enough to surpass $p^{\operatorname{win}}_{\operatorname{local}}$ in the quantum case. Since $\RGRB \in \mbox{LOC}^{2\!}({\bf PR})$ (and not in $\mbox{LOC}^{1\!}({\bf PR})$), an
optimal quantum approximation to a PR-box (known to succeed with probability $\frac{\sqrt{3}}{2}$) used instead of the perfect one only yields a $\frac{5}{2}-\sqrt{3} \approx 0.76795$ approximation to an RGB-box.

A natural question is therefore to find a quantum strategy that is better than the local one.


\newif\ifshowdiscussion
\showdiscussiontrue 
\newcommand{\discussion}[1]{\ifshowdiscussion\textcolor{red}{~{#1}}\fi}
\newcommand{\p}[4]{p_{(#1,#2|#3,#4)}} 
\newcommand{\pp}[4]{\frac{p_{(x\Equal#1|a\Equal#3)} + p_{(y\Equal #2|b\Equal #4)} + p_{(x\Equal y|#3,#4)}-1}{2} }
\newcommand{\px}[1]{p_{(x\Equal 0|#1)}}
\newcommand{\py}[1]{p_{(y\Equal 0|#1)}}
\newcommand{\pxy}[2]{p_{(x\Equal y|#1,#2)}}
\def\Equal{\texttt{=}}

\section{A better-than-local quantum strategy}\label{quantumbetter}
There is indeed a better-than-local quantum strategy that wins with probability $11/12$:

Alice and Bob share a singlet state $\ket{\psi^-}_{AB}$. According to their own input colour, they choose their measurement from the following list:
\begin{equation}
	\Pi_{\operatorname{Red}}=\ketbra{0},\Pi_{\operatorname{Green}}=\ketbra{v^-},\Pi_{\operatorname{Blue}}=\ketbra{v^+}
\end{equation}
where
\begin{equation}
	\ket{v^\pm}=\frac{1}{2}\ket{0}\pm \frac{\sqrt{3}}{2}\ket{1}
\end{equation}
These 3 projectors are located in the same plane equidistantly (like the Mercedes-Benz logo). The colour names can be permutated freely as long as Alice and Bob do the same projection for the same colour.

If the output of their measurement is positive, they output the colour that comes after their input colour in the cycle $RGB$. Otherwise, they output the previous colour. They never output their own input colour as it leads to a sure loss. 

For example, if Alice's input is Green and she measures a positive result when applying the projector $\Pi_{\operatorname{Green}}$, then $a=G$ and $x=G+1=B$ (the colour addition is modulo 3). Figure~\ref{quantumfigure} explains the protocol graphically.

\begin{figure}\centering
\%
\begin{tikzpicture}[scale=0.75]

\fill[red,opacity=0.3] (-0.2,-2.2) rectangle (0.2,2.2);
\begin{scope}[rotate around={60:(0,0)}]
\fill[green,opacity=0.3] (-0.2,-2.2) rectangle (0.2,2.2);
\end{scope}
\begin{scope}[rotate around={-60:(0,0)}]
\fill[blue,opacity=0.3] (-0.2,-2.2) rectangle (0.2,2.2);
\end{scope}

\draw[black,thick] (0,0) circle (2);

\draw[green,->,ultra thick] (0,0) -- (0,2);
\draw[blue,->,ultra thick] (0,0) -- (0,-2);

\draw[blue,->,ultra thick] (0,0) -- ({2*cos(-30)},{2*sin(-30)});
\draw[red,->,ultra thick] (0,0) -- ({-2*cos(-30)},{-2*sin(-30)});

\draw[red,->,ultra thick] (0,0) -- (-{2*cos(-30)},-{2*sin(30)});
\draw[green,->,ultra thick] (0,0) -- ({2*cos(-30)},{2*sin(30)});

\draw node[above] at (0,2.2) {$\ket{0}$};
\draw node[below] at (0,-2.2) {$\ket{1}$};
\draw node[right] at (2.2,0) {$\ket{+}$};
\draw node[left] at (-2.2,0) {$\ket{-}$};
\end{tikzpicture}
\caption{Alice and Bob's best quantum strategy is to each make the above projective measurement on their half-singlet. The basis (rectangle) depends on their own input colour. Their output is the colour of the measured arrow.}\label{quantumfigure}
\end{figure}
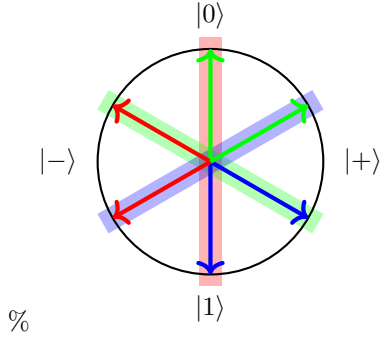
\subsection{Proof of winning probability}
We look at the probability of losing as it is simpler. To simplify notation, we call directly $x=a-1 \leftrightarrow x=0$ and $x=a+1 \leftrightarrow x=1$ as well as $y=b-1 \leftrightarrow y=0$ and $y=b+1 \leftrightarrow y=1$. Alice and Bob lose in the following cases:

\begin{equation}
	\begin{array}{ll}
		x=y  & \mbox{  if } b=a \\
		x=0\land y=1  & \mbox{  if } b=a+1 \mod3 \\
		x=1 \land y=0  & \mbox{  if } b=a-1 \mod3 \\
	\end{array}
\tag{losing cases}
\end{equation}

The probability of error $E$ only depends on the relation between $a$ and $b$ and is for each case:
\begin{align}
	E_{a=b} &=\tr \left(\ketbra{\psi^-}_{AB} \cdot \left( (\Pi_a \otimes \Pi_b) + (\Pi_a^\perp \otimes \Pi_b^\perp)\right)\right)=0\\
	E_{a+1=b}&=\tr \left(\ketbra{\psi^-}_{AB} \cdot (\Pi_a^\perp \otimes \Pi_b) \right)=\frac{1}{8}\\
	E_{a-1=b}&=\tr \left(\ketbra{\psi^-}_{AB} \cdot (\Pi_a \otimes \Pi_b^\perp) \right)=\frac{1}{8}
\end{align}
And the winning probability of this quantum strategy is (with uniformly random inputs):
\begin{equation}
p(\operatorname{win})= 1-\frac{3E_{a=b} +3 E_{a+1=b}+3E_{a-1=b}}{9} 	=\frac{11}{12}
\end{equation}

The game is therefore won with probability $11/12$ using this quantum strategy. \qedsymbol

\section{The Bell inequality associated to the RGB game}
The above quantum strategy is optimal among quantum strategies. To prove it in Section~\ref{quantumnotperfect}, we now analyze a Bell inequality associated to the RGB game. Bell game and Bell inequalities are equivalent formulations of the same phenomenon. We quickly recall how to translate from one paradigm to the other before defining the inequality and stating corresponding bounds for quantum and No-Signalling strategies.
\subsection{Bell game vs Bell inequality notations}

Up to now, we have analyzed the RGB game in the modern game context, meaning we treated strategies as probability distributions of the form $P_{x,y|a,b}$ and showed strategies in different locality classes (i.e. local, quantum or No-Signalling) can achieve different win rates. To finetune our analysis, we excluded without losing generality the output colour that always lose (i.e. $x=a$ and $y=b$) and treated the remaining outputs as binary (i.e. $0:=u-1$ and $1:=u+1$). In the next subsections, we will also use the notation $p_{(x,y|a,b)}$ for the individual conditional probabilities.

However, another way to see this problem is through Bell inequalities. Instead of looking at a game with binary outputs, one consider the properties of observables with values in $\{-1,1\}$. An observable is simply a physical quantity one can decide to measure. In physics Bell inequalities (e.g. the CHSH inequality) are usually specified by a function of the expected correlations of different observables. This function defines a quantity to which classical mechanics (i.e. local hidden variable models) imposes a limit that can be broken using quantum mechanics. We remark that all of Alice's observables need to commute (meaning the order in which they are measured don't affect their results) with all of Bob's observable to respect the No-Signalling condition common to $\mathbb{LOC}$, $|\mathbb{LOC}\rangle$ and $\mathbb{NOSIG}$.

The canonical example of a Bell inequality is the CHSH inequality. This Bell inequality also has a quantum limit: it is Tsirelson's bound. As we are about to see, this will also be the case of the RGB Bell-inequality.

The relevant point is that one can translate between the two formulations by expressing the conditional probabilities of the Bell game paradigm as expectancies of correlations in the Bell inequality paradigm, and \emph{vice versa}. We will in fact only need the following conversion equation:
\begin{align}
 \pxy{a}{b}= \frac{1+\left< A_a B_b \right>}{2} \label{corres}
\end{align}
where we noted $\left< A_a B_b \right>$ the expected correlation between the measurement outcomes of Alice's observable $A_a$ and Bob's observable $B_b$. 

\subsection{Intermediate step: rewriting the probability of winning as a function of the expected correlations between Alice's and Bob's outputs}

The following lemma will make the subsequent Bell inequality formulation simple.
\begin{lemma}\label{correlationsonly}
The probability of a given strategy distribution winning the game is given by :\begin{equation}
 p^{\operatorname{win}} = \frac{1}{9} \sum_{u=0}^2 2 - \pxy{u}{u}+\frac{\pxy{u}{u+1}}{2}+\frac{\pxy{u}{u-1}}{2}
\end{equation}
It depends only on the correlations between Alice's and Bob's output, not on their marginals.

\begin{proof}
By looking at the three losing cases above (see Section~\ref{quantumbetter}.), we obtain the probability of a distribution winning the game:
\begin{equation}
	p^{\operatorname{win}}=\frac{1}{9} \sum_{u=0}^2 \left( 3-\p{0}{0}{u}{u}- \p{1}{1}{u}{u}- \p{0}{1}{u}{u+1}-\p{1}{0}{u}{u-1}\right) \tag{winning probability equation}
\end{equation}
We rewrite it in terms of the marginals and correlations $\{\px{a},\py{b},\pxy{a}{b}\}$. Here is how we can transform each term:
\begin{align}
\p{0}{0}{a}{b}&= \frac{\px{a}+\py{b}+\pxy{a}{b}-1}{2} \\
\p{1}{1}{a}{b}&= \pxy{a}{b}- \p{0}{0}{a}{b}\\
\p{0}{1}{a}{b}&= \px{a}- \p{0}{0}{a}{b}\\
\p{1}{0}{a}{b}&= \py{b}-  \p{0}{0}{a}{b}
\end{align}
Replacing them into the winning probability equation gives:
\begin{align}
	p^{\operatorname{win}}=&\frac{1}{9} \sum_{u=0}^2  3-\p{0}{0}{u}{u}- \p{1}{1}{u}{u}- \p{0}{1}{u}{u+1}-\p{1}{0}{u}{u-1} \\
=&\frac{1}{9} \sum_{u=0}^{2} 3 - \p{0}{0}{u}{u} - \pxy{u}{u} + \p{0}{0}{u}{u} - \px{a\Equal u} + \p{0}{0}{u}{u+1} - \py{b\Equal u-1} + \p{0}{0}{u}{u-1}  \\
=&\frac{1}{9} \sum_{u=1}^2 3 - \pxy{u}{u} - \px{a\Equal u}+\pp{0}{0}{u}{u+1}\notag \\
&\hspace{0.2 \textwidth} - \py{b\Equal u} + \pp{0}{0}{u}{u-1} \\ 
 =& \frac{1}{9} \sum_{u=0}^2 2 - \pxy{u}{u}+\frac{\pxy{u}{u+1}}{2}+\frac{\pxy{u}{u-1}}{2}
\end{align}
\end{proof}
\end{lemma}

\subsection{The RGB Bell-inequality}
We show a new simple case of a Bell inequality which we call the RGB Bell-inequality.
We define it by reformulating the bound on the local winning probability of the RGB game.
\begin{proposition}
The following quantity is related to the RGB game:
 \begin{equation}
	R:= \left|\sum_{i=0}^{2} - 2\left<A_iB_i \right> +\left<A_iB_{i+1} \right>+ \left<A_iB_{i-1} \right>\right|\tag{RGB Bell-quantity}
\end{equation}
and allows us to express the RGB Bell-inequality:
\begin{equation}
	R_{\operatorname{local}}  \le 8\tag{RGB Bell-inequality}
\end{equation}

\begin{proof}
 We first rewrite the equation describing the probability of winning the RGB game into a Bell inequality notation by taking Lemma~\ref{correlationsonly} and making the simple substitution given in Eq.~\ref{corres}. We obtain:
\begin{equation}
	p^{\operatorname{win}}=\frac{1}{36} \sum_{i=0}^{2} 8 - 2\left<A_iB_i \right> +\left<A_iB_{i+1} \right>+ \left<A_iB_{i-1} \right>
\end{equation}

We then define the interesting part as the RGB Bell-inequality:
\begin{equation}
	R:=36\cdot p^{\operatorname{win}} -24=	 \left|\sum_{i=0}^{2} - 2\left<A_iB_i \right> +\left<A_iB_{i+1} \right>+ \left<A_iB_{i-1} \right>\right|
\end{equation}

Finally, from Theorem \ref{RGB89} we have $p_{\operatorname{local}}^{\operatorname{win}}\le \frac{8}{9}$, which by the last equation implies $R_{\operatorname{local}} \le 8$.

\end{proof}
\end{proposition}

As we showed in Section~\ref{quantumbetter}, quantum mechanics allows to do better than classical local strategies, but we will soon show that there is also a limit to how good quantum strategies can be. In fact, the quantum strategy we described earlier is optimal.

\begin{theorem}
The RGB Bell-inequality can be broken by quantum distributions, but there exists for the RGB game an analogue to Tsirelson's bound.
\begin{equation}
R_{\operatorname{quantum}} \le 9 \tag{quantum bound}
\end{equation}
The inequality is tight.
\begin{proof}
The value $R_{\operatorname{quantum}}=9$ is possible. It follows directly from the quantum strategy achieving a win rate of $\frac{11}{12}$, as described in Section~\ref{quantumbetter}. The proof one cannot do better is shown next in Section~\ref{quantumnotperfect}.
\end{proof}
\end{theorem}
While quantum strategies cannot reach the trivial upper bound, No-Signalling strategies can.
\begin{proposition}
No-Signalling physics (i.e. the use of the $\mathbf{R  \! \frac{GR}{BG} \! B}$-box) could break maximally the RGB Bell-inequality.
\begin{equation}
R_{\operatorname{No-Signalling}} \le 12 \tag{trivial No-Signalling bound}
\end{equation}
The inequality is tight.
\begin{proof}
The value $R_{\operatorname{No-Signalling}}=12$ is possible by using the No-Signalling strategy described in Section~1. because it achieves a win rate of $1$. The inequality is tight as all expected correlation terms are here bounded by $\{-1,1\}$.
\end{proof}
\end{proposition}

\section{Tsirelson's-like bound and proof of optimality of the quantum strategy}\label{quantumnotperfect}
We now prove the optimality of the quantum strategy described in Section~\ref{quantumbetter}. by finding a Tsirelson's-like bound for the RGB Bell-inequality.
\subsection{The optimization problem}

We want to prove that for any $\ket{\psi}$, any $\{A_a\}$ and any $\{B_b\}$, the quantum limit for the RGB Bell-inequality holds:

\begin{equation}
R_{\operatorname{quantum}}=	\left|\sum_{u=0}^{2} - 2\left<A_uB_u \right> +\left<A_uB_{u+1} \right>+ \left<A_uB_{u-1} \right>\right| \stackrel{}{\le} {9}  \tag{quantum bound}
\end{equation}
We call the value associated to our known quantum strategy $R'=9$ and the optimal value $R^*$.

\subsection{Solving the Bell inequality using semidefinite programming}
We closely follow the semidefinite programming technique explained in this article~\cite{wehner2006tsirelson}. The idea is first to transform the Bell inequality problem from the quantum realm to the vector space using a result by Tsirelson. Then we use semidefinite programming with Lagrangian duality. The key point is that the Lagrangian dual problem upper-bounds the primal problem. So by guessing a solution to the dual problem which have the same value as $R'$, we prove that $R'$ is optimal.  

\subsection{A Bell inequality as a real vector problem}
We will use an important theorem by Tsirelson\footnote{We write it as formulated in \cite{wehner2006tsirelson}, but fix a small mistake in the quantifiers order (it was correct in the original paper).}~\cite{tsirelson1980}.

\begin{theorem}[Tsirelson]
Let $A_1,\dots,A_n$ and $B_1,\dots,B_n$ be observables with eigenvalues in the interval $\{-1,1\}$.
Then for any state $\ket{\psi}\in\mathcal{A}\otimes \mathcal{B}$, there exist real unit vectors $\vec{x}_1,\dots,\vec{x}_n,\vec{y}_1,\dots,\vec{y}_n \in \mathbb{R}^{2n}$ such that for all $s,t \in \{1,\dots,n\}$:
\begin{equation}
	\bra{\psi} A_s \otimes B_t \ket{\psi} = \vec{x}_s \cdot \vec{y}_t
\end{equation}
\end{theorem}
Applying it to our case, we reduce our Bell inequality problem to maximizing the following real-vectorial expression:

\begin{equation}
R=	\sum_{i=0}^{2} - 2\vec{x}_i\cdot\vec{y}_i +\vec{x}_i\cdot\vec{y}_{i+1}+ \vec{x}_i\cdot \vec{y}_{i-1}\label{realvec} 
\end{equation}
under the constraints $\forall i, ||\vec{x}_i||=||\vec{y}_i||=1$.

\subsection{The primal problem}
We re-write the last statements in a matrix form.

\begin{equation}
	G=\begin{pmatrix}\vec{x}_1\\\vec{x}_2\\\vec{x}_3\\\vec{y}_1\\\vec{y}_2\\\vec{y}_3 \end{pmatrix}\cdot \begin{pmatrix}\vec{x}_1&\vec{x}_2&\vec{x}_3&\vec{y}_1&\vec{y}_2&\vec{y}_3\end{pmatrix} 
\end{equation}
We note $G$ can have this form if and only if it is semidefinite positive and that its diagonal elements are equal to 1 because of the normalization constraints. We also define the matrix $W$ in a way that
$\frac{1}{2} \tr GW = R_G$ where $R_G$ is the $R$ defined in Eq.~\ref{realvec} associated to this strategy $G$.
\begin{equation}
	W=\left(
\begin{array}{rrrrrr}
 0 & 0 & 0 & -2 & 1 & 1 \\
 0 & 0 & 0 & 1 & -2 & 1 \\
 0 & 0 & 0 & 1 & 1 & -2 \\
 -2 & 1 & 1 & 0 & 0 & 0 \\
 1 & -2 & 1 & 0 & 0 & 0 \\
 1 & 1 & -2 & 0 & 0 & 0 \\
\end{array}
\right)
\end{equation}

Then the semidefinite optimization primal problem is:
\begin{equation}
	\operatorname{maximize} \frac{1}{2} \tr GW \mbox{~~~subject to~~~} G\ge0 \mbox{~and~} \forall i,g_{ii}=1\tag{primal problem}
\end{equation}

\subsubsection{The primal solution}
The quantum strategy we found previously is associated with the value $R'=9$. For the sake of completeness, we prove again here this value is achievable. 
\begin{equation}
	G'=\left(
\begin{array}{rrrrrr}
 1 & -\frac{1}{2} & -\frac{1}{2} & -1 & \frac{1}{2} & \frac{1}{2} \\
 -\frac{1}{2} & 1 & -\frac{1}{2} & \frac{1}{2} & -1 & \frac{1}{2} \\
 -\frac{1}{2} & -\frac{1}{2} & 1 & \frac{1}{2} & \frac{1}{2} & -1 \\
 -1 & \frac{1}{2} & \frac{1}{2} & 1 & -\frac{1}{2} & -\frac{1}{2} \\
 \frac{1}{2} & -1 & \frac{1}{2} & -\frac{1}{2} & 1 & -\frac{1}{2} \\
 \frac{1}{2} & \frac{1}{2} & -1 & -\frac{1}{2} & -\frac{1}{2} & 1 \\
\end{array}
\right)\tag{primal solution}
\end{equation}
We check that $G'\ge0$ by looking at its eigenvalues: they are indeed $\{3,3,0,0,0,0\}$. $G'$ is therefore a feasible solution whose primal value is $9$.

\subsection{The dual problem}
We now turn to the dual problem with Lagrange multipliers.
The idea is to pose an objective function $\mathcal{L}(G,\Lambda)$ which will be equal to $R_G$ if $G$ is a feasible solution (i.e. G is semidefinite positive and all the normalization constraints are satisfied) and whose dual can be evaluated in a non-trivial way.
\begin{equation}
\mathcal{L}(G,\Lambda)=\frac{1}{2} \tr GW - \tr \Lambda (G -I_6) \tag{objective function}
\end{equation}
where $\Lambda$ is the diagonal matrix of Lagrange multipliers $\{\lambda_1,\dots,\lambda_6 \}$.
Note that $\mathcal{L}(G,\Lambda)=R_G$ for a valid solution because when the constraints are satistifed: $G-I_6=\hat{0}$.

We can associate a dual function to the objective function,
\begin{equation}
	\lambda(\Lambda)=\max_{G\text{~is~feasible}}\mathcal{L}(G,\Lambda)=\max_{G\text{~is~feasible}} \tr G(\frac{1}{2}W-\Lambda)+ \tr\Lambda \tag{dual function}
\end{equation}

The crucial fact about this dual function $\lambda(\Lambda)$ is that it upperbounds $\mathcal{L}(G,\Lambda)$, so for any feasible quantum strategy it also upperbounds $R_G$ (and therefore $R^*$). This is because \cite{convex}:

\begin{equation}
	\lambda(\Lambda)=\max_{G\text{~is~feasible}} \mathcal{L}(G,\Lambda)\ge \mathcal{L}(G^*,\Lambda)=\mathcal{L}(G^*) = R^*
\end{equation}

We therefore simply exhibit one matrix $\Lambda$ such that this upper bound $\lambda(\Lambda)$ is 9. Since we can reach it, then it will be tight.

We observe that $\lambda(\Lambda)$ evaluates to infinity if $-\frac{1}{2}W+\Lambda \not\ge0$, and that otherwise, the $G$ maximizing $\mathcal{L}(G,\Lambda)$ is the null matrix.
This leads to the following dual problem:

\begin{equation}
	\operatorname{minimize} \tr\Lambda \mbox{~~~subject to~~~} -\frac{1}{2}W+\Lambda \ge 0 \tag{dual problem}
\end{equation}

We try this solution:
\begin{equation}
	\Lambda'= \frac{3}{2} I_6\tag{dual solution}
\end{equation}

The eigenvalues of $-\frac{1}{2}W + \Lambda'$ are $\{3,3,\frac{3}{2},\frac{3}{2},0,0\}$, confirming it is semidefinite positive and thus a feasible solution (it does not lead to the trivial bound). The associated dual value is $9$ and confirms the optimality of our quantum solution.

%
%

\section{Conclusion and Open Questions}
We have defined a new game, the RGB Game, that is very simple and there exists a No-Signalling strategy winning it with probability one.
In the sense we have defined, this strategy is equivalent to the winning strategy to the PR game.
We showed the RGB game can be better won with quantum resources than classical ones but that it is still not enough to win with certainty. This is possible with No-Signalling resources.
\begin{align}
&p^{\operatorname{win}}_{\operatorname{local}}\le \frac{8}{9}&& \text{and} &\left|	\sum_{i=0}^{2} - 2\left<A_iB_i \right> +\left<A_iB_{i+1} \right>+ \left<A_iB_{i-1} \right>\right| \le 8\tag{classical local limit}\\ 
&p^{\operatorname{win}}_{\operatorname{quantum}}\le \frac{11}{12}&&\text{and}	&\left|\sum_{i=0}^{2} - 2\left<A_iB_i \right> +\left<A_iB_{i+1} \right>+ \left<A_iB_{i-1} \right>\right| \le 9\tag{quantum limit}\\ 
&p^{\operatorname{win}}_{\operatorname{No-Signalling}}\le 1&&\text{and}	&\left|\sum_{i=0}^{2} - 2\left<A_iB_i \right> +\left<A_iB_{i+1} \right>+ \left<A_iB_{i-1} \right>\right| \le 12\tag{trivial NS limit} 
\end{align}
Our main open question is whether there exist a ${\mathbb{| LOC \rangle }}$-complete distribution. Another one is to generalize all this work to distributions involving more than two parties.

\section*{Acknowledgements}
We thank
G.~
Brassard,
S.~
Fehr,
A.~
Hansen,
J.~
Li,
A.~
Magra,
A.~
Massenet,
A.~
Montina,
L.~
Salvail,
C.~
Schaffner,
S.~
Wolf
and
N.~
Yang
for discussions about early versions of this work.


\end{document}